\documentclass{llncs}
\usepackage{amsmath}
\usepackage{amssymb}
\usepackage{enumerate}
\usepackage{varioref}
\usepackage{charter,eulervm}
\usepackage{graphicx}
\usepackage{boxedminipage}

\newcommand{\es}{\varnothing}

\title{{\sc On Retracts, Absolute Retracts, and Folds in Cographs}}
\author{
 Ton~Kloks\inst{1} 
\and
Yue-Li~Wang\inst{2} 
}
\institute{ 
 Department of Computer Science\\
 National Tsing Hua University, Taiwan\\
 {\tt kloks@cs.nthu.edu.tw} 
\and
 Department of Information Management\\
 National Taiwan University of Science and Technology\\
 {\tt ylwang@cs.ntust.edu.tw}
}

\begin{document}

\maketitle

\begin{abstract}
Let $G$ and $H$ be two cographs. 
We show that the 
problem to determine whether $H$ is a retract of $G$  
is NP-complete. 
We show that this problem is fixed-parameter tractable when 
parameterized by the size of $H$. 
When restricted to the class of threshold graphs or 
to the class of 
trivially perfect graphs, the problem becomes tractable 
in polynomial time. The problem is also 
soluble in linear time 
when one cograph is given as an induced subgraph 
of the other. Foldings generalize retractions. 
We show that the problem to fold a 
trivially perfect graph onto a largest possible clique is 
NP-complete. For a threshold graph this folding number equals 
its chromatic number and achromatic number. We characterize 
the absolute retracts of cographs.  
\end{abstract}

\section{Introduction}
%%%%%%%%%%%%%%%%%%%%%%

Graph homomorphisms have regained a lot of interest by the recent
characterization of Grohe of the classes of
graphs for which Hom$(\mathcal{G},-)$ is tractable~\cite{kn:grohe}.
To be precise,
Grohe proves that, unless $FPT=W[1]$, 
deciding whether there is a homomorphism from
a graph $G \in \mathcal{G}$ to some arbitrary 
graph $H$ is polynomial if
and only if the graphs in $\mathcal{G}$ have bounded treewidth
modulo homomorphic equivalence. The treewidth of a graph modulo
homomorphic equivalence is defined as the treewidth of its core, 
ie, a minimal retract.
This, and other recent results 
(see eg~\cite{kn:bulatov2,kn:bulatov,kn:dyer,kn:goldberg,kn:grohe2,%
kn:naserasr,kn:nesetril}), 
make it desirable to have algorithms that 
compute cores, or general retracts in graphs.  

\bigskip 

For any graph $G$, all the cores of $G$ are isomorphic 
subgraphs of $G$. Therefore, one speaks of {\em the core\/} 
of a graph. However, a fixed copy of the core in $G$ 
is not necessarily a retract. Therefore, when studying 
retracts or cores one usually assumes that 
the objective is given as an induced subgraph of $G$. 
When restricted to cographs, when $H$ is given as an 
induced subgraph of $G$, it can be determined in linear time 
whether $H$ is a retract. We prove this is 
Section~\ref{section partitioned}. In the rest of the paper 
we do not assume that the graph $H$ is given as an induced 
subgraph of $G$. In that case the problem turns out to be 
NP-complete. We prove that in Section~\ref{section NP-c}.   
 
\bigskip 

In this paper we consider the retract problem
for cographs. The related surjective graph homomorphism 
problem was 
recently studied in~\cite{kn:golovach}. In this paper it was 
shown that that the problem to decide whether there is a 
surjective homomorphism from one connected 
cograph to another connected cograph 
is NP-complete.  The surjective homomorphism 
problem is also NP-complete if 
both graphs are unions of complete graphs. 
Let us mention also 
the classic result of Damaschke, 
which is that the induced subgraph isomorphism problem 
is NP-complete for cographs~\cite{kn:damaschke}. 

\bigskip 

The retract problem for cographs can be perceived 
as a pattern recognition 
problem for labeled trees. Many pattern recognition variants 
have been investigated and classified, see 
eg,~\cite{kn:chung,kn:damaschke,kn:gotz,kn:kilpelainen2,kn:kilpelainen,%
kn:matousek,kn:pinter,kn:reyner,kn:shamir,kn:sikora}. 
This last manuscript \cite{kn:sikora} 
contains references to a lot of the work done on motifs in graphs.  
However, the pattern recognition problem that corresponds with the 
retract problem on cographs seems to have eluded all these 
investigations~\cite{kn:grohe3}. 

\bigskip 

For basic terminology on graph homomorphisms we refer 
to~\cite{kn:hahn,kn:hell4}. 

\begin{definition}
Let $G$ and $H$ be graphs. A homomorphism $\phi:G \rightarrow H$ 
is a map $\phi:V(G) \rightarrow V(H)$ which preserves edges, 
that is, 
\begin{equation}
\label{eqn1}
\{x,y\} \in E(G) \quad\Rightarrow\quad \{\phi(x),\phi(y)\} \in E(H).
\end{equation}
\end{definition}
We write $G \rightarrow H$ if there is a homomorphism 
$\phi: G \rightarrow H$. 

\bigskip 

Notice that 
\begin{equation}
G \rightarrow K_k \quad\Leftrightarrow\quad \chi(G) \leq k 
\quad\text{and also that}\quad 
K_k \rightarrow G \quad\Leftrightarrow\quad \omega(G) \geq k.
\end{equation}

\bigskip 

\begin{definition}
Let $G$ and $H$ be graphs. The graph $H$ is a retract 
of $G$ if there exist homomorphisms $\rho:G \rightarrow H$ and 
$\gamma: H \rightarrow G$ such that $\rho \circ \gamma = id_H$, 
which is the identity map $V(H) \rightarrow V(H)$. 
\end{definition}
The functions $\rho$ and $\gamma$ are called the 
retraction and co-retraction, respectively. 

\bigskip 

When $H$ is a retract of $G$ then $H$ is isomorphic to an induced 
subgraph of $G$. Since there are homomorphisms in two directions, 
$G$ and $H$ have the same clique number, chromatic number and 
odd girth. 
Also, there is a retraction from $G$ to $K_k$ if and only 
if $\chi(G)=\omega(G)=k$. 

\bigskip 

There is a homomorphism $G \rightarrow H$ if and only if 
the union of $G$ and $H$ retracts to $H$. For any graph $H$, 
checking if there is a homomorphism $G \rightarrow H$ is 
polynomial when $H$ is bipartite and it is NP-complete 
otherwise~\cite{kn:hell}. 
It follows that, for any graph $H$, checking if a graph $H$ is a 
retract 
of a graph $G$ is NP-complete, unless $H$ is bipartite. 
The problem remains NP-complete, even when $H$ is 
an even cycle of length at least six, given as an induced 
subgraph of $G$~\cite{kn:feder}. The question whether a graph $G$ 
has a homomorphism to itself which is not the identity is also 
NP-complete~\cite{kn:hell3}.   
 
\bigskip 

\begin{definition}
A graph is a cograph if it has no induced $P_4$, 
which is the path with 
four vertices.
\end{definition}

Since the complement of a $P_4$ is a $P_4$, cographs are closed 
under complementation. Actually, the class of 
cographs is the smallest class 
of graphs which is closed under complementation and taking unions.  

\bigskip 

A similar characterization of cographs reads as follows. 
A graph $G$ is a cograph if and only if one of the following holds. 
\begin{enumerate}[\rm (1)]
\item $G$ has only one vertex, or 
\item $G$ is disconnected and every component is a cograph, or 
\item the complement of $G$, $\Bar{G}$ is disconnected and every 
component of $\Bar{G}$ is a cograph. 
\end{enumerate}
It follows that cographs have a decomposition tree, called a 
cotree, defined as 
follows. The decomposition tree is a rooted tree $T$. There is a 
bijection from the leaves of $T$ to the vertices of $G$. 
When $G$ has at least two 
vertices then 
each internal node of $T$, including the root, 
is labeled as $\otimes$ or $\oplus$. 
The $\oplus$ label at a node 
takes the union of the graphs that correspond 
with the children of the node. The $\otimes$ label takes the join 
of the graphs that correspond with the children.   

\bigskip 

\begin{remark}
When defined as above, the labels of the internal nodes in 
any path from the root to a leaf alternate 
between $\oplus$ and $\otimes$. 
Alternatively, one frequently defines a cotree as a rooted 
{\em binary\/} tree, 
in which each internal node is labeled as $\oplus$ and $\otimes$.   
In this paper, when talking about cotrees, we always 
assume the first type of cotree. Thus, each child of the root corresponds 
with one component or, with one cocomponent of the graph.
\end{remark}

\begin{remark}
It is well-known that cographs are recognizable in linear 
time~\cite{kn:corneil2,kn:gioan}. 
A cotree has $O(n)$ nodes, where $n=|V(G)|$, 
and it can be obtained in linear time. 
\end{remark}

\bigskip 

This paper is organized as follows. In Sections~\ref{section threshold} 
and~\ref{section TP} we show that the retract problem is polynomial 
when restricted to the classes of threshold and trivially perfect graphs. 
In Section~\ref{section NP-c} we show that the problem is 
NP-complete for cographs. In Section~\ref{section partitioned} we show 
that, when $H$ is given as an induced subgraph of $G$, it can be determined 
in polynomial time whether $H$ is a retract of $G$. In 
Section~\ref{section FPT} we show that the retract problem 
for cographs is fixed-parameter 
tractable. 
In Section~\ref{section folding} we show that computing the 
folding number is NP-complete for trivially perfect graphs. For threshold 
graphs the folding number equals the chromatic and achromatic number. 
In Section~\ref{section conclusion} we mention some 
of our ideas for future research. 

\newpage 
 
\section{Retracts in threshold graphs}
\label{section threshold}
%%%%%%%%%%%%%%%%%%%%%%%%%%%%%%%%%%%%%%

A subclass of the class of cographs is the class 
of threshold graphs. Threshold graphs are the graphs without 
induced $2K_2$, $C_4$ and $P_4$. We use the following characterization 
of threshold graphs. 
  
\begin{theorem}
A graph is a threshold graph 
if and only if every induced subgraph 
has a universal vertex or an isolated vertex. 
\end{theorem}

\begin{theorem}
\label{thm th}
Let $G$ and $H$ be threshold graphs. There exists a 
linear-time algorithm to check if $H$ is a retract of $G$. 
\end{theorem}
\begin{proof}
Assume that $H$ is a retract of $G$ and let $\rho$ and 
$\gamma$ be the retraction and co-retraction. 

\medskip 

\noindent
Assume that $G$ has a universal vertex, say $x_1$. 
Then $H$ must have a universal vertex as well, 
since a retract of a connected graph is connected. 
Let $y_1$ be a universal vertex of $H$. 
Let $y_i=\rho(x_1)$. Since $\rho$ is a homomorphism it preserves 
edges,  
and since $x_1$ is universal in $G$, $\rho$ maps no other 
vertex of $G$ to $y_i$. Notice also that $\gamma(y_i)=x_1$ 
since $\rho \circ \gamma=id_H$ and $\rho$ maps no other 
vertex to $y_i$. 

\medskip 

\noindent
Assume that $y_i \neq y_1$. Let $\gamma(y_1)=x_{\ell}$. 
Then $x_{\ell} \neq x_1$ since $\gamma$ preserves edges 
and so 
\[\{y_1,y_i\} \in E(H) \quad
\Rightarrow\quad \{\gamma(y_1),\gamma(y_i)\}=\{x_{\ell},x_1\} \in E(G) 
\quad\Rightarrow\quad x_{\ell} \neq x_1.\] 
Furthermore, 
since $y_1$ is universal, $\gamma$ maps no other vertex of $H$ to 
$x_{\ell}$. Of course, since $\rho \circ \gamma=id_H$, 
$\rho(x_{\ell})=y_1$.  

\medskip 

\noindent
We claim that $y_i$ is universal in $H$, and therefore 
exchangeable with $y_1$. 
Assume not and let $y_s \in V(H)$ be another vertex of $H$ 
not adjacent to $y_i$. Let $\gamma(y_s)=x_p$. 
Then $x_p \neq x_1$ since $\rho \circ \gamma = id_H$ and 
$\rho(x_1)=y_i \neq y_s$.  
Now, since $\rho$ is a homomorphism, 
\[\{x_1,x_p\} \in E(G) \quad\Rightarrow\quad 
\{\rho(x_1),\rho(x_p)\} = \{y_i,y_s\} \in E(H),\] 
which is a contradiction. Therefore, we may assume that $y_i=y_1$. 
\begin{center}
\begin{boxedminipage}[h]{6cm}
That is, from now on we assume that  
\[\rho(x_1)=y_1 \quad\text{and}\quad \gamma(y_1)=x_1.\]  
\end{boxedminipage}
\end{center}

\medskip 

\noindent 
This proves that, 
when $G$ is connected then $H$ is a retract of $G$ 
if and only if $H-y_1$ is a retract of $G-x_1$. 
By the way, notice that if $|V(H)|=1$ then $H$ can be a retract of $G$ 
only if $G$ is an independent set, so this case is easy to check. 

\medskip 

\noindent
Finally, assume that $G$ is not connected. Since $G$ has no induced 
$2K_2$, all components, except possibly one, have only one vertex. 
The number of components of $H$ can be at most equal to the 
number of components of $G$, since $\rho$ maps components 
in $G$ to components of $H$, and $\rho \circ \gamma=id_H$, and so 
any two components of $H$ are mapped by $\gamma$ to different 
components of $G$. 

\medskip 

\noindent
First assume that $H$ is also disconnected. 
Let $x_1,\dots,x_a$ be the isolated vertices of $G$ and 
let $y_1, \dots,y_b$ be the isolated vertices of $H$. 
Let $\rho(x_i)=y_i$ and $\gamma(y_i)=x_i$ for $i\in \{1,\dots,b\}$ 
and let $\rho(x_{b+1})=\dots=\rho(x_a)=y_b$. 
Now, $H$ is a retract of $G$ if and only if 
$H-\{x_1,\dots,x_b\}$ is a retract of $G-\{x_1,\dots,x_a\}$. 

\medskip 

\noindent
If $H$ is connected, with at least two vertices,  
then let $y_1$ be a universal vertex and let 
$\rho(x_1)=\dots=\rho(x_a)=y_1$. If 
$H$ is a retract of $G$ then $G$ must  
have exactly one component with 
at least two vertices, 
since $G$ is a threshold graph and $\rho$ is a 
homomorphism. 
Let $x_u$ be the universal 
vertex of that component and define $\rho(x_u)=y_1$ 
and $\gamma(y_1)=x_u$. In this case, $H$ is a retract 
if and only if $H-y_1$ is a retract of $G-\{x_1,\dots,x_a,x_u\}$. 

\medskip 

\noindent
An elimination ordering, which eliminates successive 
isolated and universal vertices in a threshold graph, can be 
obtained in linear time. This proves the theorem.  
\qed\end{proof}

\section{Retracts in trivially perfect graphs}
\label{section TP}
%%%%%%%%%%%%%%%%%%%%%%%%%%%%%%%%%%%%%%%%%%%%%%

\begin{definition}[\cite{kn:golumbic,kn:wolk}]
A graph $G$ is trivially perfect if for all induced 
subgraphs $H$ of $G$, $\alpha(H)$ is equal to the number 
of maximal cliques in $H$. 
\end{definition}

Trivially perfect graphs are those graphs without 
induced $C_4$ and $P_4$.   

\begin{theorem}[\cite{kn:wolk}]
A graph is trivially perfect if and only if every 
connected induced subgraph has a universal vertex. 
\end{theorem}

\begin{theorem}
\label{thm tp}
Let $G$ and $H$ be trivially perfect graphs. 
There exists an $O(N^{5/2})$ algorithm 
which checks if $H$ is a retract of $G$, where $N=|V(G)|\cdot|V(H)|$. 
\end{theorem}
\begin{proof}
Assume that $H$ is a retract of $G$. 
Let $C_1,\dots,C_t$ be the components of $G$ and let 
$D_1,\dots,D_s$ be the components of $H$. Then $s \leq t$. 
Without loss of generality, 
let $D_i$ be a retract of $C_i$ for $i \in \{1,\dots,s\}$. 
For the components $C_i$ with $i > s$, there must be a $j \leq s$ 
such that there is a homomorphism from $C_i$ to $D_j$. 

\medskip 

\noindent 
First assume that $G$ and $H$ are connected. 
Let $g_1,\dots,g_k$ be the universal vertices of $G$ and 
let $h_1,\dots,h_{\ell}$ be the universal vertices 
of $H$. As in the proof of Theorem~\ref{thm th} it follows 
that $H$ is a retract of $G$ if and only if 
\begin{enumerate}[\rm (i)]
\item $\ell \geq k$, and  
\item either $H$ is a clique and $\omega(G)=\omega(H)$ or 
$H-\{h_1,\dots,h_k\}$ is a retract of $G-\{g_1,\dots,g_k\}$. 
\end{enumerate}

\medskip 

\noindent
For the general case, consider the following bipartite graph $B$. 
The vertices of $B$ are the components of $G$ and $H$. 
There is an edge $\{C_i,D_j\} \in E(B)$ if and only if 
$C_i$ retracts to $D_j$. Then $G$ retracts to $H$ if and only if 
\begin{enumerate}[\rm (a)] 
\item $B$ has a matching which exhausts all components of $H$, and 
\item for every component $C_i$ which is not an endpoint 
of an edge in the matching 
there is a $D_j$ such that there is a 
homomorphism from $G[C_i]$ to $H[D_j]$. 
\end{enumerate}

\medskip 

\noindent
To check if a component $G[C_i]$ retracts to some $H[D_j]$ 
the algorithm greedily matches the universal vertices of 
$G[C_i]$ and $H[D_j]$ and checks 
if the remaining graph $G^{\prime}$, ie, after removal 
of the matched universal vertices, retracts to the remaining graph 
$H^{\prime}$. Let $C_i^1,\dots,C_i^{p}$ and $D_j^1,\dots,D_j^{q}$ 
be the components of $G^{\prime}$ and $H^{\prime}$. The algorithm 
constructs the bipartite graph $B_{ij}$ on the components 
$C_i^{k}$ and $D_j^{\ell}$, where $k \in \{1,\dots,p\}$ and 
$\ell \in \{1,\dots,q\}$. The algorithm checks if there is an 
edge $(C_i^k,D_j^{\ell}) \in E(B_{ij})$ 
in $O(1)$ time by table look-up, and 
so the bipartite graph $B_{ij}$ is constructed in 
\[O(pq)=O(|C_i| \cdot |D_j|).\] 
Edmonds' algorithm~\cite{kn:edmonds} 
computes a maximum matching in $B_{ij}$ in 
time 
\[O((p+q)^{5/2})=O((|C_i|+|D_j|)^{5/2}).\] 

\medskip 

\noindent
Summing over the components $C_i$ and $D_j$, for 
$i \in \{1,\dots,t\}$ and $j \in \{1,\dots,s\}$, we obtain 
\[\sum_{i=1}^t \sum_{j=1}^s |C_i|\cdot|D_j|+ (|C_i|+|D_j|)^{5/2} 
= O(|V(G)|^{5/2} \cdot |V(H)|^{5/2}).\] 

\medskip 

\noindent
This proves the claim. 
\qed\end{proof}

\section{NP-completeness of retracts in cographs}
\label{section NP-c}
%%%%%%%%%%%%%%%%%%%%%%%%%%%%%%%%%%%%%%%%%%%%%%%%%

Recall that a graph $G$ is perfect when 
$\omega(G^{\prime})=\chi(G^{\prime})$ for every 
induced subgraph $G^{\prime}$ 
of $G$. 
By the perfect graph theorem a graph is perfect if and only if 
it has no odd hole or odd antihole~\cite{kn:chudnovsky}. 
This implies that cographs are perfect.  
Perfect graphs are recognizable in polynomial 
time~\cite{kn:chudnovsky2}. 
For a graph $G$, when 
$\omega(G)=\chi(G)$ 
one can compute this value in polynomial 
time via Lov\'asz theta function~\cite{kn:grotschel}. 

The following lemma appears, eg, in~\cite{kn:fomin}. 

\begin{lemma}
\label{lm 1}
Assume that $\omega(H)=\chi(H)$. There is a homomorphism 
$G \rightarrow H$ if and only if $\chi(G) \leq \omega(H)$. 
\end{lemma}
\begin{proof}
Write $\omega=\omega(H)=\chi(H)$. 
First assume that there is a homomorphism $\phi: G \rightarrow H$. 
There is a homomorphism $f: H \rightarrow K_{\omega}$ since 
$H$ is $\omega$-colorable. 
Then $f \circ \phi: G \rightarrow K_{\omega}$ 
is a homomorphism, and so $G$ has a $\omega$-coloring. This 
implies that $\chi(G) \leq \omega$. 

\medskip 

\noindent 
Assume $\chi(G) \leq \omega$. 
There is a homomorphism $G \rightarrow K_k$, where $k = \chi(G)$. 
Since $K_k$ is an induced subgraph of $H$, there is also 
a homomorphism $K_k \rightarrow H$. 
This implies that $G$ is homomorphic 
to $H$, ie $G \rightarrow H$. 
\qed\end{proof}

\begin{corollary}
When $G$ and $H$ are perfect one can check in polynomial time 
whether there is a homomorphism $G \rightarrow H$. 
\end{corollary}

\bigskip 

It is well-known 
that retracts, like general homomorphisms, 
constitute a transitive relation. We 
provide a short proof for completeness sake. 

\begin{lemma}
Let $A$ be a retract of $G$ and let $B$ be a retract of $A$. 
Then $B$ is a retract of $G$. 
\end{lemma}
\begin{proof}
Let $\rho_1$ and $\gamma_1$ be a retraction and 
co-retraction from $G$ to $A$ 
and let $\rho_2$ and $\gamma_2$ be a retraction and 
co-retraction from $A$ to $B$. 
Since all four maps $\rho_1$, $\rho_2$, $\gamma_1$ 
and $\gamma_2$ are homomorphisms, the following two maps 
are homomorphisms as well.  
\begin{equation}
\rho_2 \circ \rho_1: G \rightarrow B \quad\text{and}\quad 
\gamma_1 \circ \gamma_2: B \rightarrow G. 
\end{equation} 
Furthermore, 
\begin{equation}
(\rho_2 \circ \rho_1) \circ (\gamma_1 \circ \gamma_2) = 
\rho_2 \circ id_A \circ \gamma_2 = \rho_2 \circ \gamma_2 = id_B.
\end{equation}
This proves that $B$ is a retract of $G$. 
\qed\end{proof}

\bigskip 

Throughout the remainder 
of this section it is assumed that $G$ and $H$ 
are cographs. Note that, using the cotree, 
$\omega(G)$ and $\chi(G)$ can be 
computed in linear time when $G$ is a cograph.  

\bigskip 

\begin{lemma}
\label{lm disc}
Assume $H$ is disconnected, with components $H_1,\dots,H_t$. 
Assume that $H$ is a retract of a graph $G$. 
Then there is an ordering 
of the components of $G$, say $G_1,\dots,G_s$ such that 
\begin{enumerate}[\rm (a)]
\item $s \geq t$, and 
\item $G_i$ retracts to $H_i$, for every $i\in\{1,\dots,t\}$, and 
\item for every $j\in \{t+1,\dots,s\}$, 
there is a homomorphism $G_j \rightarrow H$. 
\end{enumerate}
\end{lemma}
\begin{proof}
No connected graph has a disconnected retract since the homomorphic 
image of a connected graph is connected. To see that, notice 
that a homomorphism 
$\phi: G \rightarrow H$ is a vertex coloring of $G$, 
where the vertices of $H$ represent colors. By that we mean 
that, for each $v \in V(H)$, 
the pre-image $\phi^{-1}(v)$ is an independent set in $G$ or 
$\es$. 
One obtains the image $\phi(G)$ 
by identifying vertices in $G$ that receive the same color. 
When $G$ is connected, this `quotient graph' on the color classes 
is  
also connected, which is easy to prove by means of contradiction.  

\medskip 

\noindent   
Assume that $G$ retracts to $H$. 
Then we may assume that $H_1,\dots,H_t$ are induced subgraphs 
of components $G_1,\dots,G_t$ of $G$ and that  
each $G_i$ retracts to $H_i$. 
For the remaining components $G_j$, where $j > t$, there is 
then a 
homomorphisms $G_j \rightarrow H$. 

\medskip 

\noindent
Notice that, for $j > t$, 
we can check if there is a homomorphism 
$G_j \rightarrow H$ by checking if $G_j \oplus H_k$ retracts 
to $H_k$, for some $1 \leq k \leq t$ (see, eg,~\cite{kn:vikas}), 
or, equivalently (since cographs are perfect), if 
$\omega(G_j) \leq \omega(H_k)$ for some $1 \leq k \leq t$. 
\qed\end{proof}

\bigskip 

\begin{remark}
\label{remark disconnected}
Assume that we are given, for each pair $G_i$ and $H_j$ whether 
$G_i$ retracts to $H_j$ or not. Then, to check if $G$ retracts 
to $H$, we may consider a bipartite 
graph $B$ defined as follows (see eg~\cite{kn:chung,kn:reyner}). 
One color class of $B$ 
has the components of $G$ 
as vertices and the other color class has the components of $H$ as vertices. 
There is an edge between $G_i$ and $H_j$ whenever $G_i$ retracts to $H_j$. 
To check if $G$ retracts to $H$, we can let 
an algorithm compute a maximum 
matching in $B$. 
There is a retraction only if the matching exhausts all components  
of $H$ {\em and\/} if $\omega(G)=\omega(H)$. 
\end{remark}

\bigskip 

A cocomponent of a graph $G$ is a subset of vertices 
which 
induces a component of the complement $\Bar{G}$. 

\begin{lemma}
\label{lm c}
Assume $G$ is connected and assume that $G$ retracts  
to $H$. Then $H$ is also connected. Let $G_1,\dots,G_t$ be the 
subgraphs of $G$ induced by the cocomponents of $G$. Then there is a 
partition of the cocomponents of $H$ such that the subgraphs of $H$ 
induced by the parts of the partition, can be ordered $H_1,\dots,H_t$ 
such that $G_i$ retracts to $H_i$ for $i \in \{1,\dots,t\}$. 
\end{lemma}
\begin{proof}
Every subgraph $G_i$ of $G$, induced by a cocomponent,  
retracts to some induced subgraph. These retracts 
are pairwise joined, so each part is the join of 
some subgraphs induced by cocomponents of $H$.  
Thus the parts of $V(H)$ that are the images of the subgraphs 
induced by cocomponents of $G$ 
form a partition of the cocomponents of $H$. 
\qed\end{proof}

\bigskip 

\begin{figure}[htb]
\begin{center}
\includegraphics[scale=.5]{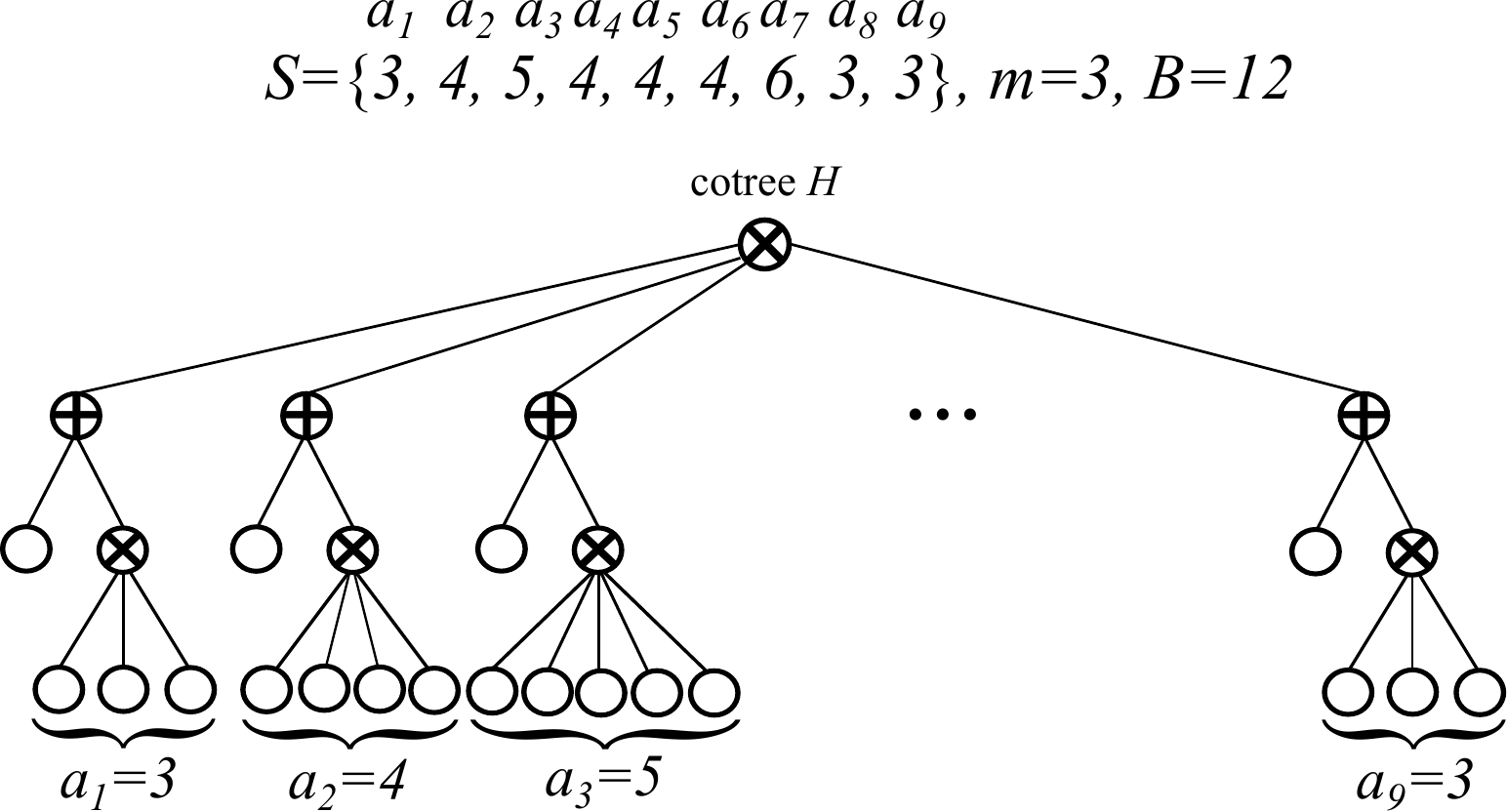}
\includegraphics[scale=.5]{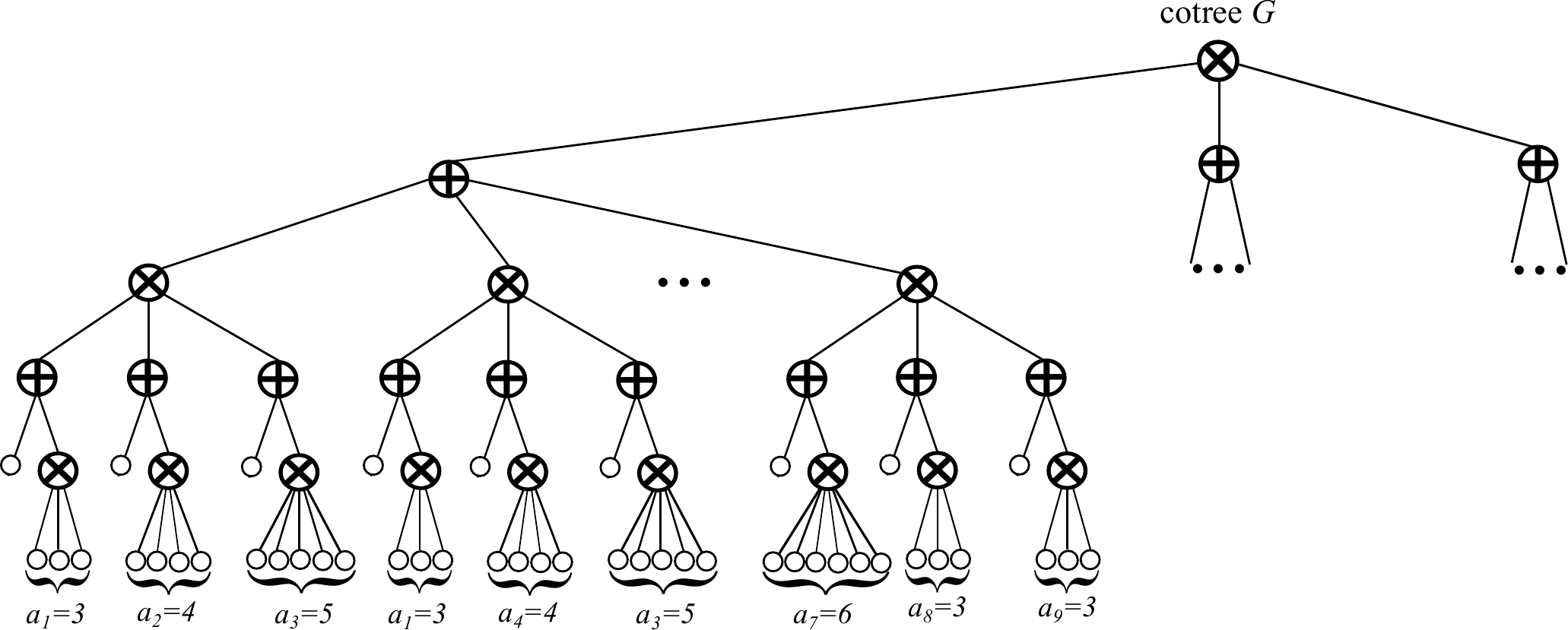}
\end{center}
\caption{The cotrees for $G$ and $H$ used in 
the proof of Theorem~\ref{thm NP-c}.}
\end{figure}

\begin{theorem}
\label{thm NP-c}
Let $G$ and $H$ be cographs. The problem to decide 
whether $H$ is a retract of $G$ is NP-complete. 
\end{theorem}
\begin{proof}
We reduce the 3-partition problem to the retract problem on cotrees. 
The 3-partition problem is the following. Let 
$m$ and $B$ be integers. Let $S$ be a multiset 
of $3m$ positive integers, $a_1,\dots,a_{3m}$. Determine if there is a 
partition of $S$ into $m$ subsets $S_1,\dots,S_m$, such that the 
sum of the numbers in each subset is $B$. Without loss of generality 
we assume that each number is strictly between $B/4$ and $B/2$, 
which guarantees that in a solution each subset contains exactly three 
numbers that add up to $B$.

\medskip 

\noindent
The 3-partition problem is strongly NP-complete, that is, the problem 
remains NP-complete when all the numbers in the input are 
represented in unary~\cite{kn:garey}. 

\medskip 

\noindent
In our reduction, the cotree for the graph $H$ has a 
root which is labeled as a join-node $\otimes$. The root 
has $3m$ children, one for each number $a_i$. For simplicity 
we refer to the children as $a_i$, $i \in \{1,\dots,3m\}$. 
Each child $a_i$ has a union node $\oplus$ as the root. 
The root of each $a_i$-child has two children, 
one is a single leaf and the 
other is a join-node $\otimes$ with $a_i$ leaves. 
This ends the description of $H$. 

\medskip 

\noindent
The cotree for the graph $G$ has a join-node $\otimes$ 
as a root and this has $m$ children. The idea is that each 
child corresponds with one set of a 3-partition of $S$. 
The subtrees for all the children are identical. 
It has a union-node $\oplus$ as the root. 
Consider all triples $\{i,j,k\}$ for which $a_i+a_j+a_k=B$. 
For each such triple create one child, which is the join of 
three cotrees, one for $a_i$, one for $a_j$ and one for 
$a_k$ in the triple. The subtree for $a_i$ is a union of 
two subtrees.  As in the cotree for the pattern $H$,  
one subtree is a single leaf, and the other subtree is the 
join of $a_i$ leaves. The other two subtrees, for the numbers 
$a_j$ and $a_k$ in the triple are similar. 

\medskip 

\noindent
Let $T_H$ and $T_G$ be the cotrees
for $H$ and $G$ as constructed above. 
Say $T_H$ and $T_G$ have roots $r_H$ and $r_G$.   
When the graph $H$ is a retract of $G$ then the $a_i$-children of 
$r_H$ are partitioned into triples, such there is a 
bijection between these triples, 
say $\{a_i,a_j,a_k\}$ and a branch in the cotree of $G$. 
Each $\oplus$-node which is the root of a child of $r_G$ 
must have 
exactly one $\{a_i,a_j,a_k\}$-child 
that corresponds with the triple. Notice that, 
by the construction, all subgraphs induced by remaining components 
of the $\oplus$-node  
have maximal cliques of size $B$. Therefore, all other 
children of the $\oplus$-node 
are homomorphic to the one child which corresponds to 
the triple $\{a_i,a_j,a_k\}$. 

\medskip 

\noindent
It now follows from Lemma~\ref{lm c} 
that there is a 3-partition if and only if 
the graph $H$ is a retract of $G$. This completes the proof.  
\qed\end{proof}

\section{The partitioned case for retracts in cographs}
\label{section partitioned}
%%%%%%%%%%%%%%%%%%%%%%%%%%%%%%%%%%%%%%%%%%%%%%%%%%%%%%%

This section is dedicated to a special 
case, the `partitioned case.' It deals with the subproblem 
where $H$ is given as an induced subgraph of $G$. 

\begin{theorem}
Let $G$ and $H$ be cographs and assume that $H$ is given as 
an induced subgraph of $G$. By that we mean that we can test in 
$O(1)$ time whether a given vertex of $G$ is in $H$ or not. 
There exists a linear-time algorithm to test if $G$ retracts 
to $H$. 
\end{theorem}
\begin{proof}
We describe the algorithm. Construct a cotree for the graph $G$. 
Repeatedly, remove children of $\oplus$-nodes for which 
\begin{enumerate}[\rm (a)]
\item the 
branch has no leaves corresponding with vertices in $H$, and 
\item the subgraph induced by the branch has 
clique number at most the clique number of a sibling. 
\end{enumerate}
When the algorithm ends such that all remaining vertices 
are in $H$ then $G$ retracts to $H$ and otherwise it does 
not. 

\medskip 

\noindent
It is easy to check the correctness of this algorithm. Also, it 
is not difficult to see that it can be implemented to run in 
linear time in the size of $G$. 
\qed\end{proof}
 
\section{A fixed-parameter solution for retracts in cographs}
\label{section FPT}
%%%%%%%%%%%%%%%%%%%%%%%%%%%%%%%%%%%%%%%%%%%%%%%%%%%%%%%%%%%%%

In this section we look at a parameterized solution for the 
retract problem. Let $G$ and $H$ be cographs. We consider the 
parameterization by the number of vertices in $H$. Let 
\[k=|V(H)|.\] 

\bigskip 

\begin{lemma}
\label{lm bound join}
When $H$ is a retract of $G$ then $\omega(G) =\omega(H) \leq k$. 
Let $T_G$ be a cotree for $G$. Then 
every join-node in $T_G$ has at most $k$ children, and the 
height of the cotree is $O(k)$. 
\end{lemma}
\begin{proof}
Let $H$ be a retract of $G$. 
By definition of a retract 
$\omega(G)= \omega(H) \leq k$.  

\medskip 

\noindent
Let $T_G$ be a cotree.  
Every $\otimes$-node $p$ in $T_G$ has at most $k$ children, 
since every child adds at least one vertex to a clique 
in the subgraph represented by $p$. 

\medskip 

\noindent
The labels on the internal nodes on a path from the root 
to the leaves alternate between $\oplus$ and $\otimes$. 
Therefore, if the 
height would be more than $2k+1$ then the path contains 
more than $k$ join nodes. Each of these join nodes adds 
at least one to a clique in $G$, which is a contradiction. 
\qed\end{proof}

\begin{theorem}
The retract problem, which asks if a cograph $H$ is 
a retract of $G$, is fixed-parameter tractable 
when parameterized by the number of vertices in $H$. 
\end{theorem}
\begin{proof}
Consider cotrees $T_G$ and $T_H$ for $G$ and $H$ and let 
$r_G$ and $r_H$ be the roots of the two cotrees. Assume both 
roots are $\otimes$-nodes. Then both have at most $k$ children. 
According to Lemma~\ref{lm c}, when $H$ is a 
retract of $G$ there is a partition $\mathcal{P}$ 
of the lines incident 
with $r_H$ such that each child of $r_G$ represents a graph 
that retracts to the subgraph of $H$ 
induced by exactly one part of the partition.  
Let $p$ be the number of children of $r_G$ and let 
$q$ be the number of children of $r_H$. The number of 
partitions of a $q$-set into $p$ nonempty parts is given by 
the Stirling number of the second kind. A trivial upperbound 
for the number of different assignments of the children of 
$r_H$ to the children of $r_G$ is 
\[p^q \leq k^k,\] 
since, by Lemma~\ref{lm bound join}, $p \leq k$ and $q \leq k$. 

\medskip 

\noindent
Our algorithm tries all possible partitions of the children 
of $r_H$. Consider a partition $\mathcal{P}$, and let 
$P_i \in \mathcal{P}$ be mapped to the $i^{\mathrm{th}}$ child 
of $r_G$. Let $H_i$ be the subgraph of $H$ induced by $P_i$ 
and let $G_i$ be the cocomponent of $G$ induced by the $i^{\mathrm{th}}$ 
child of $r_G$. We proceed as in the proof of Theorem~\ref{thm tp}. 
Let $C_i^1,\dots,C_i^a$ be the components of the root of the 
$i^{\mathrm{th}}$ child of $r_G$. Let $D_i^1,\dots,D_i^b$ be the 
components of $H_i$. Consider the bipartite graph with vertices 
the components of $G_i$ and $H_i$, where an edge 
$(C_i^{\alpha},D_i^{\beta})$ indicates that $C_i^{\alpha}$ 
retracts to $D_i^{\beta}$. The algorithm checks if there 
is a matching that exhausts all components of $H_i$, and it 
checks if the remaining components of $G_i$ are homomorphic to 
some component of $H_i$. 

\medskip 

\noindent
Since the height of the cotree is bounded by $k$, it follows 
that this algorithm can be implemented to run in 
$O(k^{k^2} \cdot (k\cdot|V(G)|)^{5/2})$. This proves the theorem.  
\qed\end{proof}

\section{Foldings}
\label{section folding}
%%%%%%%%%%%%%%%%%%

\begin{definition}
Let $G=(V,E)$ be a graph and let $x$ and $y$ be two 
vertices in $G$ that are at distance two. A simple fold with 
respect to $x$ and $y$ is the operation which identifies $x$ 
and $y$. A folding is a homomorphisms which is a sequence 
of simple folds. 
\end{definition}
When $G \rightarrow H$ is a folding then we say that $G$ folds 
onto $H$. 

It is well-known that any retraction is a folding, 
see eg~\cite[Proposition~2.19]{kn:hahn}. 

\begin{definition}
The folding number $\Sigma(G)$ 
of a connected graph $G$ is the largest number $s$ 
such that $G$ folds onto $K_s$. When $G$ is disconnected 
the folding number is the maximal folding number of the 
graphs induced by the components of $G$.  
\end{definition}

Recall that the achromatic number $\Psi(G)$ of a graph $G$ 
is the largest number of colors with which one can properly color 
the vertices of $G$ such that for any two colors there are two 
adjacent vertices that have those colors. 

\begin{lemma}
\label{lm fold}
Assume that $G$ has a universal vertex $u$. Then 
\[\Sigma(G)=1+\Psi(G-u)=\Psi(G).\]
\end{lemma}
\begin{proof}
Any two nonadjacent vertices of $G-u$ are at distance two 
in $G$. Thus any achromatic coloring of $G$ is a folding. 
The universal vertex must be in a color class by itself. 
Harary and Hedetniemi~\cite{kn:harary} show that, when $G$ 
is the join of two graphs $G_1$ and $G_2$ then 
$\Psi(G)=\Psi(G_1)+\Psi(G_2)$. 
The proves the lemma. 
\qed\end{proof}

\bigskip 

Notice that the achromatic number problem is NP-complete,  
even for trees. However, the problem is 
fixed-parameter tractable~\cite{kn:farber,kn:manlove,kn:mate}. 
The image of a tree after a simple fold is 
a tree. Therefore, the folding number of a tree is at most two. 

\begin{theorem}
The problem to compute the folding number is NP-complete, 
even when restricted to trivially perfect graphs. 
\end{theorem}
\begin{proof}
Bodlaender shows in~\cite{kn:bodlaender} that computing the 
achromatic number is NP-complete, even when restricted 
to trivially perfect graphs. 
Since the class of trivially perfect graphs is closed under 
adding universal vertices, by Lemma~\ref{lm fold} computing the 
folding number is NP-complete for trivially perfect graphs. 
\qed\end{proof}

\begin{theorem}
When $G$ is a threshold graph then 
\[\chi(G)=\Sigma(G)=\Psi(G).\]
\end{theorem}
\begin{proof}
When $G$ is the join of two graph $G_1$ and $G_2$ then 
\[\Psi(G)=\Psi(G_1)+\Psi(G_2).\] 
Assume that $G$ has an isolated vertex $x$. In any achromatic 
coloring, the vertex must have a color that is used by another 
vertex also. Therefore, 
\[\Psi(G)=\max \;\{\;1,\Psi(G-x)\;\}.\] 
This proves the theorem. 
\qed\end{proof}

\section{Absolute retracts for cographs}
\label{section absolute retracts}
%%%%%%%%%%%%%%%%%%%%%%%%%%%%%%%%%%%%%%%%

\begin{definition}
Let $\mathcal{G}$ be a class of graphs.
A graph $H$ is an absolute retract for $\mathcal{G}$ if
$H$ is a retract of a graph $G \in \mathcal{G}$ whenever
$G$ is an isometric embedding of $H$ and $\chi(H)=\chi(G)$.
\end{definition}

Hell, in his PhD thesis, characterized absolute retracts for the
class of bipartite graphs as the retracts of components of categorical
products of paths~\cite{kn:hell}. Pesch and Poguntke characterized
absolute retracts of $k$-chromatic graphs~\cite{kn:pesch}.
Their characterization can be strengthened for the case of
bipartite graphs such that it leads to a polynomial
recognition algorithm for absolute retracts of
bipartite graphs~\cite{kn:bandelt}. Examples
of absolute retracts of bipartite graphs are the chordal
bipartite graphs~\cite{kn:golumbic2}.
Median graphs are exactly the absolute retracts of
hypercubes~\cite{kn:bandelt2}. For reasons of brevity we leave
out the mention of all results on reflexive graphs.

\bigskip 

To illustrate the
problem for cographs,
let $G$ be the butterfly, ie $G=K_1 \otimes (K_2 \oplus K_2)$, and
let $H$ be the paw, ie $H=K_1 \otimes (K_1 \oplus K_2)$.
The $H$ is an induced subgraph of $G$ and, obviously, $H$ has an
isometric embedding in $G$ and $\omega(G)=\omega(H)$.
But $H$ is not a retract of $G$ and so
$H$ is not an absolute retract for cographs.

\begin{theorem}
Let $H$ be a connected cograph. Then $H$ is an
absolute retract for the class of cographs if and only
if every vertex of $H$ is in a maximal clique
of cardinality $\omega(H)$.
\end{theorem}
\begin{proof}
First notice that a cograph $G$ is an isometric
embedding of a connected cograph $H$ if and only if
$H$ is an induced subgraph of $G$. This follows
since connected cographs have diameter two or, also,
because they are distance hereditary. The observation
is true for distance-hereditary graphs
simply by definition~\cite{kn:howorka}.

\medskip

\noindent
Let $H$ be a connected cograph. Write $\omega=\omega(H)$
and assume that every vertex of $H$ is in a clique of
cardinality $\omega$. Let $G$ be a cograph with
$\omega(G)=\omega$ such that $H$ is an induced subgraph of
$G$.

\noindent
First assume that $G$ is disconnected. Then the
vertices of $H$ are contained in one component of $G$ since
$H$ is connected. If $W$ is any other component, then
$G[W]$ has clique number at most $\omega$ and so there is a
homomorphism from this component to the component that
contains $H$. In other words, $H$ is a retract of $G$
if and only if $H$ is a retract of the component that
contains $H$ as an induced subgraph.
Henceforth, we may assume that $G$ is connected.

\medskip

\noindent
Consider a cotree for $G$. Since $G$ is connected the root
is an $\otimes$-node. Let
\[C_1,\dots,C_t\] be the
cocomponents of $G$.
Since $H$ is an induced subgraph with
the same cliquenumber as $G$, $H$ decomposes into
the same number of cocomponents $D_1,\dots,D_t$.
Notice that
\[\omega=\sum_{i=1}^t \omega(G[C_i])=\sum_{i=1}^t \omega(H[D_i]).\]
Therefore, since $\omega(H[D_i]) \leq \omega(G[C_i])$, we have
equality for each $i$, that is,
\[\omega(H[D_i])=\omega(G[C_i]) 
\quad\text{for $i \in \{1,\dots,t\}$.}\]

\medskip

\noindent
Now consider an $\oplus$-node. Let $C_1^{\prime},\dots,C_{\ell}^{\prime}$
be the sets of vertices of the
subgraphs of $G$ induced by the children.
Let
\[\omega^{\prime}= \max \;\{\; \omega(G[C_i^{\prime}])\;|\;
i \in \{1,\dots,\ell\}\;\}.\]
Write
\[D_i^{\prime}=V(H) \cap C_i^{\prime} 
\quad\text{for $i \in \{1,\dots,\ell\}$.}\]
Since $\omega(G)=\omega(H)$, we have that
there is at least one component $C_i^{\prime}$
such that
\[\omega^{\prime}=\omega(H[D_i^{\prime}]).\]

\medskip

 \noindent
For $G$ to retract to $H$ we must have that,
for every $j \in \{1,\dots,\ell\}$,
\[D^{\prime}_j \neq \es \quad\text{implies}\quad 
\omega(H[D_j^{\prime}]) = \omega(G[C_j^{\prime}]).\]
This condition is satisfied by virtue of the condition that every vertex of
$H$ is in a clique of cardinality $\omega$. Namely, this
implies that
for every $j \in \{1,\dots,\ell\}$,
\[D^{\prime}_j \neq \es \quad\text{implies}\quad 
\omega(H[D^{\prime}_j])=\omega^{\prime}.\]

\medskip

\noindent
Notice that this condition is necessary for
$H$ to be an absolute retract. This can be
seen as follows.
If there were a component $D_j$ with
\begin{equation}
\label{eq1}
\omega(H[D^{\prime}_j]) < 
\max \; \{\; \omega(H[D^{\prime}_i])\;|\; i \in \{1,\dots,\ell\}\;\}
\end{equation}
then we could construct a cograph $G$ such that $H$
is an induced subgraph of $G$ with $\omega(G)=\omega(H)$
but such that $G$ does not contract to $H$. Namely, add
one vertex to a component $D^{\prime}_j$ satisfying~\eqref{eq1}
as a true twin of a vertex which is in a maximum clique
of $H[D^{\prime}_j]$.

\medskip

\noindent
This proves the theorem.
\qed\end{proof}
 
\section{Concluding remarks}
\label{section conclusion}
%%%%%%%%%%%%%%%%%%%%%%%%%%%%

One interesting problem that we leave open is whether the folding number is 
fixed-parameter tractable. 

\bigskip 

We proved that the retract problem for cographs is 
fixed-parameter tractable when parameterized by the number 
of vertices in the smaller of the two graphs. Perhaps more 
challenging and useful would be the `cleaning-parameter,' 
ie, the {\em difference\/} between the number of vertices 
of the two graphs, 
recently introduced by Marx and Schlotter~\cite{kn:marx1,kn:marx2}. 
This parameterization was investigated for the induced 
subgraph isomorphism problem, when restricted to various 
classes of graphs, eg, interval graphs, trees and planar graphs, 
and grids.  
As far as we know, whether cographs can be cleaned by a 
fixed-parameter algorithm is an open problem.


\begin{thebibliography}{99}
%%%%%%%%%%%%%%%%%%%%%%%%%%%

\bibitem{kn:bandelt2}Bandelt,~H., Retracts of hypercubes,
{\em Journal of Graph Theory\/} {\bf 8} (1984), pp.~501--510.

\bibitem{kn:bandelt}Bandelt,~H., A.~D\"ahlmann and H.~Sch\"utte, 
Absolute retracts of bipartite graphs, 
{\em Discrete Applied Mathematics\/} {\bf 16} (1987), pp.~191--215. 

\bibitem{kn:bodlaender}Bodlaender,~H., 
Achromatic number is NP-complete for cographs and interval graphs, 
{\em Information Processing Letters\/} {\bf 31} (1989), pp.~135--138. 

\bibitem{kn:bulatov2}Bulatov,~A., 
The complexity of the counting constraint satisfaction problem, 
{\em Proceedings of the $35^{\mathrm{th}}$ International 
Colloquium on Automata, Languages and Programming\/}, 
Springer-Verlag, LNCS~5125 (2008), pp.~646--661. 

\bibitem{kn:bulatov}Bulatov,~A., V.~Dalmau, M.~Grohe and D.~Marx, 
Enumerating homomorphisms, 
{\em Journal of Computer and System Sciences\/} {\bf 78} (2012), 
pp.~638--650. 

\bibitem{kn:courcelle}Courcelle,~B. and S.~Oum, 
Vertex-minors, monadic second-order logic, and a conjecture 
by Seese, {\em Journal of Combinatorial Theory, Series B\/} {\bf 97}, 
(2007), pp.~91--126. 

\bibitem{kn:chudnovsky2}Chudnovsky,~M., G.~Cornu\'ejols, 
X.~Liu, P.~Seymour and K.~Vu\v{s}kovi\'c, 
Recognizing Berge graphs, 
{\em Combinatorica\/} {\bf 25} (2005), pp.~143--186. 

\bibitem{kn:chudnovsky}Chudnovsky,~M., N.~Robertson, 
P.~Seymour and R.~Thomas, 
The strong perfect graph theorem, 
{\em Annals of Mathematics\/} {\bf 164} (2006), 
pp.~51--229. 

\bibitem{kn:chung}Chung,~M., 
$O(n^{2.5})$ time algorithms for the subgraph homomorphism problem 
on trees, 
{\em Journal of Algorithms\/} {\bf 8} (1987), pp.~106--112. 

\bibitem{kn:corneil}Corneil,~D., H.~Lerchs and L.~Stewart-Burlingham, 
Completent reducible graphs, 
{\em Discrete Applied Mathematics\/} {\bf 3} (1981), pp.~163--174. 

\bibitem{kn:corneil2}Corneil,~D., Y.~Perl and L.~stewart, 
A linear recognition algorithm for cographs, 
{\em SIAM Journal on Computing\/} {\bf 14} (1985), pp.~926--934. 

\bibitem{kn:damaschke}Damaschke,~P., 
Induced subgraph isomorphism for cographs is NP-complete, 
{\em Proceedings $16^{\mathrm{th}}$ WG'90\/}, 
Springer-Verlag, LNCS~484 (1991), pp.~72--78. 

\bibitem{kn:dyer}Dyer,~M. and D.~Richerby, 
An effective dichotomy for the counting constraint satisfaction problem. 
Manuscript on ArXiv: 1003.3879, 2011. 

\bibitem{kn:edmonds}Edmonds,~J., 
Paths, trees, and flowers, 
{\em Canadian Journal of Mathematics\/} {\bf 17} (1965), 
pp.~449--467. 

\bibitem{kn:farber}Farber,~M., G.~Hahn, P.~Hell and D.~Miller, 
Concerning the achromatic number of graphs, 
{\em Journal of Combinatorial Theory, Series B\/} {\bf 40} (1986), 
pp.~21--39. 
 
\bibitem{kn:feder}Feder,~T., P.~Hell and J.~Huang, 
List homomorphisms and circular arc graphs, 
{\em Combinatorica\/} {\bf 19} (1999), pp.~487--505. 

\bibitem{kn:feder2}Feder,~T., P.~Hell, P.~Johnsson, A.~Krokhin 
and G.~Nordh, 
Retractions to pseudoforests, 
{\em SIAM Journal on Discrete Mathematics\/} {\bf 24} (2010), 
pp.~101--112. 

\bibitem{kn:fomin}Fomin,~F., P.~Heggernes and D.~Kratsch, 
Exact algorithms for graph homomorphisms, 
{\em Theory of Computing Systems\/} {\bf 41} (2007), pp.~381--393. 

\bibitem{kn:garey}Garey,~M. and D.~Johnson, 
{\em Computers and intractability: a guide to the theory of 
NP-completeness\/}, W.~H.~Freeman and co., 1979. 

\bibitem{kn:gioan}Gioan,~E. and C.~Paul, 
Split-decomposition and graph-labelled trees: characterizations and 
fully dynamic algorithms for totally decomposable graphs. 
Manuscript on ArXiv: 0810.1823, 2008. 

\bibitem{kn:goldberg}Goldberg,~L., M.~Grohe, M.~Jerrum and 
M.~Thurley, 
A complexity dichotomy for partition functions with mixed signs, 
{\em SIAM Journal on Computing\/} {\bf 39} (2010), pp.~3336--3402. 

\bibitem{kn:golovach}Golovach,~P., B.~Lidick\'y, B.~Martin and 
D.~Paulusma, 
Finding vertex-surjective graph homomorphisms, 
{\em Acta Informatica\/} {\em 49} (2012), pp.~381--394. 

\bibitem{kn:golumbic}Golumbic,~M., 
Trivially perfect graphs, 
{\em Discrete Mathematics\/} {\bf 24} (1978), pp.~105--107. 

\bibitem{kn:golumbic2}Golumbic,~M. and C.~Goss,
Perfect elimination and chordal bipartite graphs,
{\em Journal of Graph Theory\/} {\bf 2} (1978), pp.~155--163.

\bibitem{kn:gotz}G\"otz,~M., C.~Koch and W.~Martens, 
Efficient algorithms for the tree homeomorphism problem, 
{\em Proceedings of the $11^{\mathrm{th}}$ International 
Conference on Database Programming Languages\/}, Springer-Verlag, 
LNCS~4797 (2007), pp.~17--31. 

\bibitem{kn:grohe}Grohe,~M., 
The complexity of homomorphism and constraint satisfaction 
problems seen from the other side, 
{\em Journal of the ACM\/}, {\bf 54} (2007). 

\bibitem{kn:grohe3}Grohe,~M., 
Personal communication. 

\bibitem{kn:grohe2}Grohe,~M., K.~Kawarabayashi, D.~Marx and 
P.~Wollan, 
Finding topological subgraphs is fixed-parameter tractable, 
{\em STOC'11, Proceedings of the $43^{\mathrm{rd}}$ Annual ACM 
Symposium on Theory of Computing\/}, ACM (2011), pp.~479--488.
 
\bibitem{kn:grotschel}Gr\"otschel,~M., L.~Lov\'asz and A.~Schrijver, 
Polynomial algorithms for perfect graphs. 
In (Berge, Chv\'atal eds.): 
{\em Topics on perfect graphs\/}, 
North-Holland, Mathematics Studies {\bf 88} (1984), pp.~325--256. 

\bibitem{kn:hahn}Hahn,~G. and C.~Tardif, 
Graph homomorphisms: structure and symmetry. 
In (G.~Hahn and G.~Sabidussi eds.) 
{\em Graph symmetry -- algebraic methods and applications\/}, 
NATO ASI Series C: Mathematical and Physical Sciences, Vol.~497, 
Kluwer, 1997, pp.~107--166. 

\bibitem{kn:harary}Harary,~F. and S.~Hedetniemi, 
The achromatic number of a graph, 
{\em Journal of Combinatorial Theory\/} {\bf 8} (1970), 
pp.~154--161. 

\bibitem{kn:hell5}Hell,~P., 
{\em R\'etractions de graphes\/}. PhD Thesis, Universit\'e de 
Montr\'eal, 1972. 
 
\bibitem{kn:hell}Hell,~P. and J.~Ne\v{s}et\v{r}il, 
On the complexity of $H$-coloring, 
{\em Journal of Combinatorial Theory, Series B\/} {\bf 48} (1990), 
pp.~92--110. 

\bibitem{kn:hell3}Hell,~P. and J.~Ne\v{s}et\v{r}il, 
The core of a graph, 
{\em Discrete Mathematics\/} {\bf 109} (1992), pp.~117--126. 

\bibitem{kn:hell4}Hell,~P. and J.~Ne\v{s}et\v{r}il, 
{\em Graphs and homomorphisms\/}, 
Oxford Universty Press, 2004. 
\bibitem{kn:howorka}Howorka,~E.,

A characterization of distance-hereditary graphs,
{\em The Quarterly Journal of Mathematics\/} {\bf 28} (1977),
pp.~417--420.

\bibitem{kn:imrich}Imrich,~W. and S.~Klav\v{z}ar, 
Retracts of strong products of graphs, 
{\em Discrete Mathematics\/} {\bf 109} (1992), pp.~147--154. 

\bibitem{kn:kilpelainen2}Kilpel\"ainen,~P. and H.~Mannila, 
Quey primitives for tree-structured data, 
{\em Proceedings of the $5^{\mathrm{th}}$ Annual Symposium on 
Combinatorial Pattern Matching\/}, Springer-Verlag, LNCS~807 
(1994), pp.~213--225. 
 
\bibitem{kn:kilpelainen}Kilpel\"ainen,~P. and H.~Mannila, 
Ordered and unordered tree inclusion, 
{\em SIAM Journal on Computing\/} {\bf 24} (1995), pp.~340--356. 

\bibitem{kn:klavzar}Klav\v{z}ar,~S., 
Absolute retracts of splitgraphs, 
{\em Discrete Mathematics\/} {\bf 134} (1994), pp.~75--84. 

\bibitem{kn:kruskal}Kruskal,~J., 
Well-quasi-ordering, the tree theorem, and Vazsonyi's conjecture, 
{\em Transactions of the American Mathematical Society\/} 
{\bf 95} (1960), 
pp.~210--225. 

\bibitem{kn:loten}Loten,~C., 
{\em Retractions of chordal and related graphs\/}. 
PhD Thesis, Simon Fraser University, 2003. 

\bibitem{kn:manlove}Manlove,~D. and C.~McDiarmid, 
The complexity of harmonious coloring for trees, 
{\em Discrete Applied Mathematics\/} {\bf 57} (1995), pp.~133--144. 
 
\bibitem{kn:marx1}Marx,~D. and I.~Schlotter, 
Parameterized graph cleaning problems, 
{\em Discrete Applied Mathematics\/} {\bf 157} (2009), 
pp.~3258--3267. 

\bibitem{kn:marx2}Marx,~D. and I.~Schlotter, 
Cleaning interval graphs, 
{\em Algorithmica\/} {\bf 65} (2013), pp.~275--316. 

\bibitem{kn:mate}M\'at\'e,~A., 
A lower estimate for the achromatic number of irreducible graphs, 
{\em Discrete Mathematics\/} {\bf 33} (1981), pp.~171--183. 

\bibitem{kn:matousek}Matou\v{s}ek and R.~Thomas, 
On the complexity of finding iso- and other morphisms 
for partial $k$-trees, 
{\em Discrete Mathematics\/} {\bf 108} (1992), pp.~343--364. 

\bibitem{kn:naserasr}Naserasr,~R. and Y.~Nigussie, 
On a new reformulation of Hadwiger's conjecture, 
{\em Discrete Mathematics\/} {\bf 306} (2006), pp.~3136--3139. 

\bibitem{kn:nesetril}Ne\v{s}et\v{r}il,~J. and P.~de~Mendez, 
Cuts and bounds, 
{\em Discrete Mathematics\/} {\bf 302} (2005), pp.~211--224. 

\bibitem{kn:pesch}Pesch,~E. and W.~Poguntke, 
A characterization of absolute retracts of $n$-chromatic graphs, 
{\em Discrete Mathematics\/} {\bf 57} (1985), pp.~99--104. 

\bibitem{kn:pinter}Pinter,~R., O.~Rokhlenko, D.~Tsur 
and M.~Ziv-Ukelson, 
Approximate labelled subtree homeomorphism, 
{\em Journal of Discrete Algorithms\/} {\bf 6} (2008), 
pp.~480--496. 

\bibitem{kn:reyner}Reyner,~S., 
An analysis of good algorithms for the subtree problem, 
{\em SIAM Journal on Computing\/} {\bf 6} (1977), 
pp.~730--732. 

\bibitem{kn:shamir}Shamir,~R. and D.~Tsur, 
Faster subtree isomorphism, 
{\em Journal of Algorithms\/} {\bf 33} (1999), pp.~267--280. 

\bibitem{kn:sikora}Sikora,~F., 
An (almost complete) state of the art around the graph motif 
problem. Manuscript 2012. 

\bibitem{kn:thomasse}Thomass\'e,~S., 
On better-quasi-ordering countable series-parallel orders, 
{\em Transactions of the American Mathematical Society\/} {\bf 352} 
(2000), pp.~2491--2505. 

\bibitem{kn:vikas}Vikas,~N., 
A complete and equal computational complexity classification 
of compaction and retraction to all graphs with at most four 
vertices and some general results, 
{\em Journal of Computer and System Sciences\/} {\bf 71} 
(2005), pp.~406--439. 

\bibitem{kn:wahlstrom}Wahlstr\"om,~M., 
New plain-exponential time classes for graph homomorphism, 
{\em Theory of Computing Systems\/} {\bf 49} (2011), pp.~273--282. 

\bibitem{kn:wolk}Wolk,~E., 
A note on ``The comparability graph of a tree,'' 
{\em Proceedings of the American Mathematical Society\/} 
{\bf 16} (1965), pp.~17--20. 

\end{thebibliography}
\end{document}